\documentclass{article}
\usepackage[utf8]{inputenc}
\usepackage{indentfirst}
\usepackage{graphicx,tikz}
\usetikzlibrary{decorations.pathreplacing}
\usetikzlibrary{arrows}
\usepackage{fancyhdr}
\usepackage{hyperref}
\hypersetup{pdftitle={Loops},pdfauthor={Lorenzo Panebianco}}
\graphicspath{ {images/} }
\usepackage[T1]{fontenc}
\usepackage{mathrsfs}
\usepackage[english]{babel}
\usepackage{float,floatflt, epsfig}
\usepackage{subfig}
\usepackage{color}
\usepackage{caption}
\usepackage{upgreek}
\usepackage{fancyhdr}
\usepackage{lipsum} 
\usepackage{latexsym}
\usepackage{mathtools}
\usepackage{mathtools}
\DeclarePairedDelimiter\ceil{\lceil}{\rceil}
\DeclarePairedDelimiter\floor{\lfloor}{\rfloor}
\usepackage{amssymb,amsthm,amsmath}
\usepackage{bbm}
\usepackage{amscd,xypic}
\usepackage{mathrsfs}
\usepackage{braket}
\usepackage[toc,page]{appendix}
\usepackage{enumerate} 
\usepackage{booktabs}
\usepackage{mathtools}
\usepackage[all]{xy}
\usepackage[margin=1.1in]{geometry}

\newcommand{\g}{\mathfrak{g}}

\newcommand{\C}{\mathbb{C}}
\newcommand{\T}{\mathbb{T}}
\newcommand{\N}{\mathbb{N}}

\newcommand{\R}{\mathbb{R}}
\newcommand{\Z}{\mathbb{Z}}
\newcommand{\Diff}{\textit{\em Diff}_+(S^1)}

\newcommand{\hi}{\mathcal{H}}

\newcommand{\M}{\mathcal{M}}

\theoremstyle{plain}
\newtheorem{thm}{Theorem}

\newtheorem*{thm*}{Theorem}
\newtheorem{cor}[thm]{Corollary}
\newtheorem*{cor*}{Corollary}

\newtheorem{lem}[thm]{Lemma} 
\newtheorem*{lem2a*}{Lemma 2A} 
\newtheorem*{lem2b*}{Lemma 2B} 
\newtheorem*{lem2c*}{Lemma 2C} 
\newtheorem{prop}[thm]{Proposition}

\newtheorem{example}{Example}

\newtheorem*{claim*}{Claim} 

\theoremstyle{definition}
\newtheorem{defn}[thm]{Definition}

\newtheorem*{defn*}{Definition}

\theoremstyle{remark}
\newtheorem{remark}[thm]{Remark}

\title{Loop groups and QNEC}
\author{Lorenzo Panebianco* \\ \\
	*Dipartimento di Matematica,  Università di Roma “La Sapienza” \\
	Piazzale Aldo Moro 5, 00185–Roma, Italy \\
	E-mail: panebianco@mat.uniroma1.it}
\date{November 2020}

\begin{document}
		
		\maketitle
		
		\begin{abstract}
We investigate some analytical properties of loop group models, showing that a Positive Energy Representation (PER) of a loop group $LG$ can be extended to a PER of $H^{3/2}(S^1,G)$ for any compact, simple and simply connected Lie group $G$. We then explicitly compute the adjoint action of $H^{5/2}(S^1,G)$ on the stress energy tensor and we use these results to prove the Quantum Null Energy Condition (QNEC) and the Bekenstein Bound for states obtained by applying a Sobolev loop to the vacuum. We also give a simpler proof of these last results in the case $G=SU(n)$. Finally, we construct and study solitonic representations of the loop group conformal nets induced by the conjugation by a loop with a discontinuity in $-1$.
		\end{abstract}
		
		{\bf Keywords:} Loop Group, Positive Energy Representation, Relative Entropy
		
		{\bf MSC2020:} primary 81T05, 81T40

\section{Introduction}

Recently, much attention has been focused on quantum information aspects of Quantum Field Theory, which naturally takes place in the framework of quantum black holes thermodynamics. However, more unexpected and interesting connections between the relative entropy and the stress energy tensor have arisen, and in particular it is of interest to provide and prove an axiomatic formulation of the Quantum Null Energy Condition (QNEC). A general proof of the QNEC is given in \cite{CF}.  While in \cite{LP} we study the QNEC on the Virasoro net, in this work we focus on loop group models. \\

We mainly follow \cite{TL}. Let $G$ be a compact, simple and simply connected Lie group. A {\em Positive Energy Representation (PER)} of the loop group $LG=C^\infty(S^1, G)$ on a separable Hilbert space $\hi$ is a projective strongly continuous unitary representation $\pi$ of $	{LG} \rtimes \T$ with a commutative diagram
\[
\xymatrix{
	{LG} \rtimes \T  \ar@{->}[r]^{\pi} 
	& PU(\hi)   \\
	\T \ar@{->}[u]  \ar@{->}[r]^{R} 
	& U(\hi)  \ar@{->}[u] 
}
\]
where the torus $\T \cong \text{Rot} $ acts on $LG$ by rotations and $R$  is a strongly continuous unitary representation which induces an isotypical decomposition $\hi = \bigoplus_{n \geq 0} \hi(n)$. In the following we will consider only PERs of {\em finite type}, namely we require $ \dim \hi(n) < + \infty $ for every $n$. We also require, without loss of generality, that $\hi(0)$ is not zero-dimensional. Irreducible PERs are of finite type. \\

We denote by $\mathfrak{g}_0$ the Lie algebra of $G$  and by $\mathfrak{g}$ the complexification of $\mathfrak{g}_0$. Recall that $\mathfrak{g}_0$ is a compact Lie algebra, that is its Killing form is negative definite. In particular, there is an antilinear involution $x \mapsto x^*$ of $\g$ such that 
\[
\g_0 = \{ x \in \g \colon x^* = -x \} \,.
\] 
We define an involution of $L^\text{pol}\g$ by $X(n)^*=X^*(-n)$, where $X(n) = X e^{ in \theta } $ for $X$ in $\g$. We denote by $\hi^{\text{fin}}$ the subspace of finite energy vectors, namely the algebraic sum of the subspaces $\hi(n)$, and by $\hi^\infty \subseteq \hi$ the Fréchet space  of smooth vectors. For $X $ in $ L \mathfrak{g} $ and $\gamma$ in $LG$, on $\hi^\infty$ we have that $	[d, \pi(X)]  = i \pi(\dot{X}) $ and that $\pi(X)^*  =  \pi(X^*) $. Moreover, we have the commutation relations
\[
	[\pi(X), \pi(Y)]  = \pi([X, Y]) + i \ell B(X, Y) \,, \qquad B(X, Y)  =  \int_0^{2 \pi} \braket{X, \dot{Y}} \frac{d \theta}{2 \pi} \,.
\]
The adjoint action of $LG$ is given by
\begin{equation}
	\begin{split}
	\pi(\gamma) \pi(X) \pi(\gamma)^* & = \pi(\gamma X \gamma^{-1}) + i c(\gamma, X)  \,,  \\
		\pi(\gamma) d \pi(\gamma)^* & = d - i \pi(\dot{\gamma} \gamma^{-1}) + c(\gamma,d)  \,,
	\end{split}
\end{equation}
where $d$ is the generator of rotations, namely $\pi(R_\theta)= e^{i \theta d}$.  The real constants $c(\gamma, X)$ and $c(\gamma, d)$ are explicitly given by
\[
c(\gamma, X)  = - \ell \int_0 ^{2 \pi} \braket{\gamma^{-1} \dot{\gamma}, \dot{X}} \frac{d \theta}{2 \pi} \,, \qquad c(\gamma,d)  =- \frac{\ell}{2} \int_0 ^{2 \pi} \braket{\gamma^{-1} \dot{\gamma}, \gamma^{-1} \dot{\gamma}} \frac{d \theta}{2 \pi} \,.
\]

Here $\braket{\cdot , \cdot}$ denotes the basic inner product, namely the Killing form normalized on the highest root $\theta$ in such a way that $\braket{\theta, \theta} =2$. Notice that the map $X \mapsto \pi(X)$ is actually a representation if restricted on $\g$. Indeed, the projective representation of $G$ lifts to a unitary representation and the subspaces $\hi(n)$ are $G$-invariant. We will use the following notation:
\[
x= \pi(X)   \,, \quad x(n) = \pi(X(n))\,, \quad \braket{x, y } = \braket{X, Y }   \,.
\]
We can define a representation of the Virasoro algebra $\text{Vir}$ 
\[
[L_n, L_m] = (n-m)L_{n+m} +\delta_{n+m,0} \frac{n(n^2-1)}{12} c \,,
\]
by Sugawara construction, that is such a representation is given by defining
\[
L_n = \frac{1}{2(\ell + g)} \sum_{m } \colon x_i(-m)x^i(m+n) \colon \,,
\]
where we used the Einstein convention on summations.  Here $\{x_i \}$ and $\{x^i \}$ are dual basis with respect to the basic inner product, namely $\braket{x^i, x_j} = \delta_{ij}$, and $g$ is the {\em dual Coxeter number}, that is
\[
g= 1 + \sum {a}^\vee_i  \,, \quad \theta = \sum {a}^\vee_i \alpha^\vee_i \,,
\]
where $\alpha^\vee_i$ are the simply coroots and ${a}^\vee_i $ are strictly positive. By the assumption $\braket{\theta, \theta} =2$ it can be shown that  the dual Coxeter number is half the Casimir of the adjoint representation, namely we have $[X_i, [X^i,Y]] = 2g Y $ for $Y$ in $\mathfrak{g} $. Notice that if $X_i$ belongs to $\g_0$ then $x_i(n)^* = - x_i(-n)$. If the PER is irreducible, then the central charge $c$ and the trace anomaly $h$ are given by
\[
c = \frac{\ell \dim \g }{\ell + g} \,, \quad h = \frac{C_\lambda}{2(\ell + g)} \,,
\]
where $C_\lambda$ is the Casimir associated to the basic inner product $\braket{\cdot, \cdot}$ and to the null energy space $\hi(0)= \hi_\lambda$, which is the irreducible highest weight representation of $\g$ associated to some dominant integral weight $\lambda$ satisfying
\begin{equation} \label{eq:alcove}
	\braket{\lambda, \theta } \leq \ell \,.
\end{equation}
The set of dominant integral weights $\lambda$ satisfying condition \eqref{eq:alcove} is called the {\em level $\ell$ alcove}. We will say that $\pi$ is a {\em vacuum positive energy representation}, or simply a {\em vacuum representation}, if $\hi(0) $ is one-dimensional. If $\hi(0) = \C \Omega$ with $(\Omega | \Omega)=1$, then the state $\omega$ associated to $\Omega$ is called the {\em vacuum state}. Notice that $\pi$ is a vacuum representation if and only if $h=0$. More in general, if $\hi(0)= \hi_\lambda$ then the trace anomaly can be computed by taking in account that
\[
C_\lambda = \braket{\lambda, \lambda + 2 \rho } \,, \quad g = 1 + \braket{\rho,  \theta }\,,
\]
where $\rho$ is the {\em Weyl vector}, that is the sum of all the dominant integral weights. Equivalently, the Weyl vector can be defined as half the sum of all the positive roots. \\

Now we briefly study the unitary irreducible representations $V(c,h) $ of the Virasoro algebra appearing from an irreducible level $\ell$ positive energy representation of $LG$. If $\ell =0$ then $\lambda =0$, and by $c=h=0$ we have the trivial representation of $\text{Vir}$. If $\ell \geq 1$, then $V(c,h) $ belongs to the {\em continuous series}, namely we have $h \geq 0$ and $c \geq 1$. The estimate on the central charge follows by the inequality $g+1 \leq \dim \g $, which can be noticed by studying the following table: 

\begin{center}
\begin{tabular}{|l|l|l|l|l|l|l|l|l|l|}
	\hline
	Dynkin diagram  & $A_n$  &  $B_n$  &  $C_n$  &  $D_n$ & $E_6$ & $E_7$ & $E_8$ & $F_4$ & $G_2$      \\
	\hline
	Complex simple Lie algebra  & $\mathfrak{sl}_{n+1}$  &  $\mathfrak{so}_{2n+1}$  &  $\mathfrak{sp}_{2n}$  &  $\mathfrak{so}_{2n}$ & $\mathfrak{e}_{6}$ & $\mathfrak{e}_{7}$ & $\mathfrak{e}_{8}$ & $\mathfrak{f}_{4}$ & $\mathfrak{g}_{2}$      \\
	\hline
	Complex dimension  & $n^2 + 2n$   &  $2n^2 + n$  &  $2n^2 + n$  &  $2n^2 - n$   & $78$  & $133$ & $248$ & $52$ & $14$    \\
	\hline
	Dual Coxeter number & $n+1$      &  $2n-1$  &  $n+1$  &  $2n-2$ & $12$   & $18$   & $30$ & $9$ & $4$     \\
	\hline
\end{tabular} 
\end{center}



\begin{lem} \label{lem:comm} \cite{GW}  $[L_n, x(k)]=-kx(n+k)$ on $\hi^\text{\em fin}$. 
\end{lem}
\begin{proof}
	Let $\{x_i \}$ be a orthonormal basis in $\g_0$, so that its dual basis is given by $x^i=-x_i$.  We have the following relations:
	\begin{equation*}
		\begin{split}
		 [xy,z] & = x[y,z]+[x,z]y \,, \\
			[x(a), [y,z](b)] & = [[z,x],y](a+b) + [[x,y](a),z(b)] \,, \\
			[x_i(a), [x^i,z](b)] & = [x_i, [x^i,z]](a+b) =2g z(a+b ) \,,
	\end{split}
	\end{equation*}
	for any integers $a$ and $b$. Notice also that on  $\hi^{\text{fin}}$ we have
	\begin{equation*}
		\begin{split}
			L_n & = \frac{1}{2(\ell + g)} \sum_{m \geq -n/2 } (2- \delta_{-m, n/2}) x_i(-m)x^i(m+n) \,.
		\end{split}
	\end{equation*}
Therefore we can compute
\begin{equation*}
\begin{split}
[x_i(-m)x^i(m+n), x_j(k)] & = x_i(-m)[x^i(m+n), x_j(k)] + [x_i(-m),x_j(k)]x^i(m+n)  \\
& = - \ell k \delta_{i,j}(\delta_{k,-m-n} + \delta_{k,m})x_j(n+k) \\ 
& \qquad + x^i(-m)[x_i,x_j](m+n+k) + [x_i,x_j](-m+k)  x^i(m+n) \,.
\end{split}
\end{equation*}
The structure constants $\{c^h_{ij} \}$ relative to $\{x_i \}$, that is the constants determined by $ [x_i, x_j] = c^h_{ij}x_h $, verify the relations $c^h_{ij} = c^i_{jh} = - c^i_{hj}$. It follows that
\begin{equation*}
\begin{split}
[x_i, x_j](a)x^i(b)  & = -x^i(a)[x_i, x_j](b) \,, \\
x^i(a)[x_i, x_j](b) + x^i(b)  [x_i, x_j](a)  & = 2g x_j(a+b) \,, \\
x^i(a)[x_i, x_j](a)  & = g x_j(2a) \,.
\end{split}
\end{equation*}
In particular, if we set $ X_m = x^i(-m)[x_i,x_j](m+n+k) $ then $ X_m + X_{-m-n-k} = 2 g x_j (n+k)  $. Thus, by explicit computation one can prove that
\begin{equation*}
\begin{split}
[L_n, x_j(k) ] & =  -\frac{k \ell}{(\ell + g)} x_j(n+k) + \frac{1}{2(\ell + g)} \sum_{m \geq -n/2 } (2- \delta_{-m, n/2}) (X_m - X_{m-k}) \\
& = - k x_j(n+k) \,,
\end{split}
\end{equation*}
and therefore we have $[L_n, x(k) ] =- k x(n+k)$ on $\hi^\text{fin}$ for every $x$ in $\g$ and $k,n$ in $\Z$.  
\end{proof}

As a corollary, the representation of $\text{Vir} = \C \cdot c \oplus \partial $, with $\partial$ the Witt algebra, extends to a representation of the semidirect product $\g[t,t^{-1}]\rtimes \text{Vir} \cong \tilde{\g} \rtimes \partial$, with $\tilde{\g} = \g[t, t^{-1}] \oplus \C \cdot c$. Indeed, if we set $L_n = \pi (\ell_n)$, where $ \ell_n(\theta) = e^{in\theta} \frac{d}{d \theta} $, then we can define the stress energy tensor $\pi(h) = \sum_n \hat{h}_n L_n $ for any polynomial vector field $h$ on the circle, namely a vector field which is a finite linear combination of the fields $ \ell_n$. Therefore, by Lemma \ref{lem:comm}  we have $ [\pi(h), \pi(X)] = \pi(h.X)  $ on $\hi^\text{fin}$ for every $X$ in $L^{\text{pol}}\g$, where $h.X(\theta) = h(\theta) \frac{d}{d \theta}X(\theta)$. \\

The space $L^{\text{pol}}\g$ can be completed to a Banach Lie algebra $L\g_t $, with $t \geq 0$. Indeed, given $X = \sum_k a_k e^{ik \theta}$ in $L^{\text{pol}}\g$, we define $L\g_t $ as the completion of $L^{\text{pol}}\g$ with respect to the norm
\[
|X|_t = \sum_k (1 + |k|)^t \Vert a_k \Vert \,.
\]
We have norm continuous embeddings with dense range $ C^{\ceil*{t}+1} (S^1, \g)\hookrightarrow L\g_t \hookrightarrow C^{\floor{t}} (S^1, \g) 
$, and for any $t \geq n$ we have $ \Vert X^{(n)}\Vert_\infty \leq |X|_t $.  Notice that, in general, we can similarly define the space  $L\g_{s,p} $ by
\[
|X|_{s,p} = \bigg(  \sum_k (1 + |k|)^{sp} \Vert a_k \Vert^p  \bigg)^{1/p} \,.
\]
We set $\mathcal{S}_t = L\C_t$, namely the space of continuous complex functions $h$ on $S^1$ satisfying
\[
|h|_t=\sum_k (1 + |k|)^t \Vert \hat{h}_k \Vert < + \infty \,.
\]
Notice that by Fourier expansion we can naturally identify $\mathcal{S}_t $  with a space of Sobolev vector fields on the circle. Finally, we define two actions of $\mathcal{S}_t $  by $h.X(\theta) = h(\theta) \frac{d}{d \theta}X(\theta)$ and $hX(\theta) = h(\theta)X(\theta)$. Indeed, by noticing that $ (1+|n+k|)^t \leq (1+|n|)^t(1+|k|)^t $ we have  $|h.X|_{s,p}  \leq |h|_s |X|_{s+1,p} $ and $|hX|_{s,p}  \leq |h|_s |X|_{s,p} $.

\section{Sobolev loop groups}

We know that $LG= C^\infty(S^1, G)$ is a Fréchet Lie group if endowed with the Whitney smooth topology. Its topology is induced by the norms defined on the Banach Lie groups $ L^kG = C^{k}(S^1, G) $. The exponential map $\exp_{LG} \colon L\g_0 \to LG $ is naturally defined by $ \exp_{LG}(X) = \exp_G \cdot X $ and is a local homeomorphism near the identity (see \cite{PS}). Here we define and describe some properties of Sobolev loop groups. \\

Let $M$ be a riemannian  manifold. Suppose $M$ to be isometrically embedded in $\R^\nu$ for some  $\nu>0$. Define, for $1 \leq p < \infty$ and $0 \leq s < \infty$, 
\[
W^{s,p}(S^1,M) = \{ f \in W^{s,p}(S^1, \R^\nu) \colon f(\theta) \in M \, \text{\em  a.e.} \} \,.
\]
Here $W^{s,p}(S^1, \R^\nu)$ is the completion of $C^\infty(S^1, \R^\nu)$ with respect to the norm $ \Vert f \Vert_{s,p} =  \Vert \Delta^{s/2} f \Vert_p + \Vert f \Vert_p 
$, where $\Delta \geq 0$ is the smallest closure on $L^p(S^1, \R^\nu)$ of the laplacian seen as an operator on $C^\infty(S^1, \R^\nu)$.  \\

In the following, every compact Lie group $G$ will be considered as a riemannian Lie group with respect to the unique riemannian structure extending $-\braket{\cdot , \cdot}$, namely the opposite of the basic inner product, and such that left and right translations are smooth isometries. We show that if $G$ is compact and simple then every faithful unitary representation $\rho \colon G \to U(n)$  induces an isometrical embedding of $G$ in some real euclidean space. By continuity of the representation we have that $G$ is represented as a compact embedded Lie subgroup of    $U(n)$. Moreover, by simplicity of $\g_0$ we have that $ \lambda \text{tr}(\rho(x)^*\rho(y)) = - \braket{x , y}$ for some $\lambda > 0$. Therefore if we consider $M_n(\C)$ as a real vector space with inner product $\lambda \text{Re} \, \text{tr}(X^*Y)$ then we have an isometric embedding $G \hookrightarrow M_n(\C)$. \\

\begin{thm}
	If $G$ is a compact, simple and simply connected Lie group faithfully represented in some space of matrices, then $W^{s,p}(S^1, G)$ is an analytic Banach Lie group for $1 < p, sp < \infty$ whose Lie algebra is $W^{s,p}(S^1, \g_0)$. Moreover,  $C^\infty(S^1, G)$ is dense in $W^{s,p}(S^1, G)$ and thus $W^{s,p}(S^1, G)$ is connected.
\end{thm}
\begin{proof}
	First we show that $W^{s,p}(S^1, G)$ is a topological group. This can be proved by using the fact that any two functions $f,g$ in  $W^{s,p}(S^1, \R^\nu )$ verify, for $p $ and $sp$ in $(1, \infty)$, the estimate \cite{SA}
	\[
	\Vert fg \Vert_{s,p} \leq C_{s,p} \Vert f \Vert_{s,p} \Vert g \Vert_{s,p} \,.
	\]
By this estimate and by the identity $f^{-1} - g^{-1} = f^{-1}(g-f)g^{-1}$ it follows that  $W^{s,p}(S^1, G)$ is a topological group for $p$ and $sp$ in $(1, \infty)$, since it is clearly a Hausdorff space. Now we define the map
	\[
	\exp_{s,p} \colon W^{s,p}(S^1, \g_0) \to W^{s,p}(S^1, G) \,, \quad \exp_{s,p}(X)(z) = \exp_G(X(z))  \,.
	\]
This map is a local homeomorphism since $\exp_G \cdot X$ is an absolutely convergent series for $\Vert X \Vert_{s,p} < 1/C_{s,p}$. We check that $W^{s,p}(S^1, G)$ is connected. By the density of $C^\infty(S^1, G)$ in $W^{s,p}(S^1, G)$ (see Theorem 1.1. of  \cite{BPVS}), it suffices to prove that $C^\infty(S^1,G)$ is path connected and then connected. But a smooth homotopy between two loops in $G$ is a path in $C^\infty(S^1,G)$ and the connectedness follows. Finally, we conclude if we prove that the group operations of inversion and multiplication are analytic. By connectedness we can reduce to prove this in an open neighborhood of the identity (see \cite{TL}, Lemma 2.2.1.). The inversion $X \mapsto -X$ is clearly analytic. The analyticity of  left and right multiplication follows from the Baker-Campbell-Hausdorff-Dynkin formula, and the theorem is proved.
\end{proof}

\begin{cor} \label{cor:prod}
	Every loop $\gamma$ in $W^{s,p}(S^1, G)$ is a finite product of exponentials.
\end{cor}

We have formally defined our Sobolev loop group $W^{s,p}(S^1, G)$ and we have checked that such a space has good topological and analytical properties. Now we are finally ready to extend our positive energy representation of $LG$.

\begin{lem} \label{lem:net}
	Let $X$ be a topological space. Consider  a net $(x^i)_{i \in I}$ in $X$ such that $x^i \to x$ for some $x$ in $X$. Suppose that for every $i \in I$ there is a directed set $A_i$ and a net $(x^i_\alpha)_{\alpha \in A_i}$ convergent to $x^i$. Then there is a net with values in $D= \{ x^i_\alpha \}$ which is convergent to $x$.
\end{lem}
\begin{proof}
 Fix a neighborhood $U$ of $x$ and an element ${\alpha}_i$ in $A_i$ for some $i \in I$. By convergence one can find an element ${\alpha}'_j$ in $A_j$ for some $j \geq i$, with eventually ${\alpha}_i \leq {\alpha}'_j$ if $i=j$, which belongs to $U$. This means that  $x$ is an accumulation point of the net
	\[
	A = \coprod_{i \in I} A_i \to D \, ,\quad \alpha_i \mapsto y^i_{\alpha_i} \,.
	\]
	Here $A $ is the disjoint union of the sets $A_i$ considered with the lexicographical order:  $\alpha _i \leq \alpha'_j$ if and only if $i \leq j$ in $I$ and $\alpha _i \leq \alpha'_i$ in $A_i$ in the case $i=j$. It follows the existence of a subnet $B \to D$ of $A$ convergent to $x$.
\end{proof}

\begin{example} \label{example:net}
Let $G$ be a topological group. Let $(x_\alpha)_{\alpha \in A}$ and $(y_\beta)_{\beta \in B}$ be nets in $G$, with $x_\alpha \to x$ and $y_\beta \to y$. Then the net $(x_\alpha y_\beta)_{ ( \alpha, \beta ) \in A \times B }$, where $A \times B$ is the product directed set, is convergent to $xy$. Indeed,	given a neighborhood $W$ of the identity, there is a neighborhood $V$ of the identity such that $V^2 \subseteq W$. By definition there exist $\alpha_0$ in $A$ and $\beta_0$ in $B$ such that $x_\alpha$ is in $xV$ and $y_\beta$ is in $Vy$ for every $\alpha \geq \alpha_0$ and $\beta \geq \beta_0$, and by $xV^2 y \subseteq xWy$ the assertion follows.
\end{example}

\begin{prop} \label{prop:extension}
	Let $i \colon G \to H$ and $\pi \colon G \to U$ be two homomorphisms of topological groups. We suppose $H$ to be connected and $i(G)$ to be dense in $H$. We suppose the existence of a neighborhood $V$ in $H$ of the identity such that
	\[
	p(v) =\lim_\alpha \pi(g_\alpha)\,, \quad i(g_\alpha) \to v \,,
	\] 
	is well defined and does not depend on the the choice of the net $(g_\alpha)_{\alpha \in A}$. Then, we can define a homomorphism $p \colon H \to U$ such that $\pi = p \cdot i$.
\end{prop}
\begin{proof}
	By the connectedness of $H$ we have that $H = \cup_n V^n$. We show by induction that $p$ is well defined on $V^n$ for every $n$. Suppose the thesis true for $V^n$, and consider elements $w$ in $V^n$ and $v $ in $V$. Pick a net $h_\beta$ such that $ i(h_\beta) \to v$. By inductive hypothesis and Lemma \ref{lem:net} we have that
\begin{equation}
	 \lim_\alpha  \pi(i(g_\alpha))= \lim_\beta \lim_\alpha  \pi(i(g_\alpha h^{-1}_\beta)) \pi(v)
\end{equation}
is well defined and does not depend on the net $g_\alpha$ such that $i(g_\alpha) \to wv$. Hence $p$ is well defined and clearly $\pi = p \cdot i$. The continuity follows by Lemma \ref{lem:net}, and Example \ref{example:net} shows that $p$ is actually a group homomorphism.
\end{proof}

We now consider the Fréchet Lie group $\Diff$. The action of $\Diff$ on $LG$ is smooth and any irreducible positive energy representation of $LG$ is $\Diff$-covariant. 	Let $h(\theta) \frac{d}{d \theta}$ be a smooth vector field on the circle, or more generally in $(L\g_0)_{3/2}$. We recall that the stress energy tensor
\[
T(h) = \sum_n \hat{h}_n L_n \,,
\]
is an essentially self-adjoint operator with $\hi^{\text{fin}}$ as a dense core. 
For any $X + i \alpha h$ in $L \g \rtimes i \R h$ we define $\pi(X+ i\alpha h) = \pi(X) + i\alpha  T(h)$. We notice that by the bound
\[
\Vert (1 + L_0)^k L_n \xi \Vert \leq \sqrt{c/2} (1+|n|)^{k + 3/2} \Vert (1+L_0)^{k+1} \xi \Vert \,, \quad k \in \N \,,
\]
we have
\[ \Vert \pi(X+  ih) \xi \Vert_k
\leq \sqrt{2\kappa} |X|_{k+1/2} \Vert \xi \Vert_{k + 1/2} + \sqrt{c/2} | {h} |_{k+3/2} \Vert \xi \Vert_{k+1}  \,.
\]

\begin{prop} \label{prop:sob}
    If $X$ is in $W^{s,p}(S^1, \g_0)$ for $1 \leq p \leq 2$ and $s>3/2 + 1/p$, then $\pi(X)$ is a closable operator which is essentially skew-adjoint on any core of $d$.
\end{prop}
\begin{proof}
    We know that for every vector $\xi$ in $\hi^\text{fin}$ and $X$ in $L^\text{pol}\g$ we have \cite{TL}
	\begin{equation}
		\begin{split}
		\Vert \pi(X) \xi \Vert_t & \leq \sqrt{2(\ell + g)} |X|_{|t|+ 1/2} \Vert \xi \Vert_{t + 1/2} \,, \\
		\Vert [1+d,\pi(X)] \xi \Vert_t & \leq \sqrt{2(\ell + g)} |X|_{|t|+ 3/2} \Vert \xi \Vert_{t+ 1/2} \,,
		\end{split}
	\end{equation}
	for any $t$ in $\R$. By density one extends $\pi$ to  $(L\g)_{|t|+3/2}$ in such a way to still verify the same estimates for $\xi$ in $\hi^{t+ 1/2}$. It follows that if $X$ is in $L\g_{3/2}$  then both $\pi(X)$ and $ [1+d,\pi(X)] $ are bounded  operators from $\hi^{1/2}$ to $\hi^0 \subseteq \hi^{-1/2}$. Therefore, by the Nelson commutator theorem (Thm. X.36 in \cite{RS}) we have that if $X$ is in $(L\g_0)_{3/2}$ then the restriction of $\pi(X)$ on
	\[
	\mathcal{D}= \{ \psi \in \hi \cap \hi^{1/2} \colon \pi(X) \in \hi \}
	\]
	is a closable operator on $\hi$ which is essentially skew-adjoint on any core of $d$ such as $\hi^\text{fin}$. 	Notice now that, by standard arguments, there is a norm continuous embedding $W^{s,p}(S^1, \g) \hookrightarrow L\g_{3/2} $. Indeed, if $X (\theta)= \sum_k a_k  e^{i k \theta}$  then by the Hölder inequality
	\begin{equation}
	| X |_{3/2}  = \sum_k (1 + |k|)^{3/2} \Vert a_k \Vert = \sum_k (1 + |k|)^{3/2-s} (1 + |k|)^{s} \Vert a_k \Vert \leq A_{s,p} |{X} |_{s,p'} \leq B_{s,p} \Vert {X} \Vert_{s,p} \,,
	\end{equation}
	where $A_{s,p} $ exists finite by construction and $B_{s,p} $ exists finite by Riesz-Thorin. Therefore, by the arguments given above we have that if $X$ is in $W^{s,p}(S^1, \g_0)$ then $\pi(X)$ is a skew-symmetric operator on $\hi^\textit{\em fin}$ which is essentially skew-adjoint on any core of $d$.
\end{proof}

Propositions \ref{prop:extension} and \ref{prop:sob} can be used to extend a positive energy representation of $LG$ to a strongly continuous projective representation of $W^{s,p}(S^1, G)$. However, for simplicity and convenience in the following we will focus on $H^s(S^1, G)=W^{s,2}(S^1, G)$.

\begin{prop}
    The map $\pi$ of $L\g$ can be extended to $H^{3/2}(S^1, \g)$, with $\pi(X)$ closable and such that
    \begin{equation} \label{eq:H2-est}
    \Vert \pi(X) \xi \Vert \leq C \Vert X \Vert_{H^{3/2}} \Vert \xi \Vert_{1/2} \,, \quad \xi \in \hi^{1/2} \,.
    \end{equation}
  Moreover, $ \pi(X)^* = \overline{\pi(X^*)} $, and in particular $ \pi(X) $ is essentially skew-adjoint if $X$ is in $ H^{3/2}(S^1, \g_0) $. 
\end{prop}
\begin{proof}
    We use some techniques shown in \cite{CW}. Consider $X= \sum_k a_k e^{i k \theta}$ in $H^{3/2}(S^1, \g)$. By
    \begin{equation}
    	\begin{split}
    	|X|_{1/2} & = \sum_k (1  + |k| )^{1/2} \Vert a_k \Vert  \leq  \bigg( \sum_k (1+ |k|)^{-2} \bigg)^{1/2}  \bigg( \sum_k (1+ |k|)^3 \Vert a_k \Vert^2 \bigg)^{1/2} ,
    	\end{split}
    \end{equation}
    we have that $\pi(X)$ is well defined and \eqref{eq:H2-est} follows by the previous estimates. It is also closable since $\pi(X^*) \subseteq \pi(X)^*$. Notice that since $\hi^\text{fin}$ is a core for $(1+d)^{1/2}$ then $\pi(X^*)$ is the formal adjoint of $\pi(X)$ on $\hi^{1/2}$ for any $X$ in $H^{3/2}(S^1, \g)$. Now we define on $\hi^{1/2}$ the operator
    \[
    R_{X, \epsilon} = [\pi(X), e^{- \epsilon d}] \,,
    \]
    which is well defined since $e^{- \epsilon d} \colon \hi \to \hi^\infty \subseteq \hi^{1/2}$. Since $ - R_{X^*, \epsilon} \subseteq R^*_{X, \epsilon}$ then $R_{X, \epsilon}$ is closable. Notice that if $d v_k = k v_k$ then
    \[
    R_{x(n), \epsilon} v_k = f_{n,k}(\epsilon) x(n) v_k \,, \quad f_{n,k}(\epsilon)= e^{- \epsilon k}- e^{-\epsilon(k-n)} \,.
    \]
    By simple analysis techniques one can prove that 
    \[
    |f_{n,k+n}(\epsilon)|^2(1+k+n) \leq 2 (1 + |n|)^2 \,,
    \] 
    for any $ \epsilon \geq 0 $, $k \geq 0$ and $n+ k \geq 0$. Therefore if $v = \sum_{k \geq 0} v_k$ is in $\hi^\text{fin}$ then for every $X=\sum_j x_j(n_j)$ in $L^\text{pol} \g$ we have
    \begin{equation}
    	\begin{split}
    	\Vert R_{X, \epsilon }v \Vert^2 & = \sum_{k} \Vert (R_{X, \epsilon }v)_k \Vert^2  = \sum_{j,k} |f_{n_j,k+n_j}(\epsilon)|^2 \Vert x_j(n_j) v_{k+n_j} \Vert^2 \\
    	& \leq 2(\ell + g) \sum_{j,k} |f_{n_j,k+n_j}(\epsilon)|^2 (1 + |n_j|) (1 + k + n_j)  \Vert x_j \Vert^2 \Vert  v_{k+n_j} \Vert^2 \\
     	& \leq 4(\ell + g) \sum_{j,k}  (1 + |n_j|)^3  \Vert x_j \Vert^2 \Vert  v_{k+n_j} \Vert^2 \\
     	& \leq  4(\ell + g) \Vert  X \Vert_{H^{3/2}}^2  \Vert  v \Vert^2 \,.
    	\end{split}
    \end{equation}
    In particular, $R_{X,\epsilon}$ can be defined as a bounded operator for every $X$ in $H^{3/2}(S^1, \g)$ and $R_{X,\epsilon} \to 0$ strongly as $\epsilon \to 0$. Moreover, by the identity $R^*_{X,\epsilon} = - \overline{R_{X^*, \epsilon}} $ we have that $R^*_{X,\epsilon} \to 0$ strongly as well. Now we arrive to the crucial point: if $v$ is in $\mathcal{D}(\pi(X)^*)$ then 
    \[
    \pi(X^*) e^{- \epsilon d}v =      \pi(X)^* e^{- \epsilon d}v =  e^{- \epsilon d} \pi(X)^*v - R^*_{X, \epsilon} v \to \pi(X)^*v\,, \quad \epsilon \to 0 \,,
    \]
    and this concludes the proof since $e^{- \epsilon d}v  \to v$.
\end{proof}

\begin{thm} \label{thm:ext}
	Let $\pi \colon LG \to PU(\hi)$ be a positive energy representation of $LG$. Then $\pi$ can be extended to a positive energy representation of $H^{3/2}(S^1,G)$. 
\end{thm}
\begin{proof}
We consider an open neighborhood $U$ in $H^{3/2}(S^1, \g_0)$ on which the exponential map of $H^{3/2}(S^1, G)$ is a homeomorphism and set $V = \exp_{H^{3/2}}(U)$. For $\gamma = \exp_{H^{3/2}}(X)$ in $V$ we define in $PU(\hi)$
	\[
	\pi(\gamma) = e^{\pi(X)} \,, \quad  X \in U \,.
	\]
The neighborhood $V$ verifies Proposition \ref{prop:extension}, since if $\gamma_\alpha= \exp(X_\alpha)$ converges to $\gamma= \exp(X)$ in $V$ then the estimate \eqref{eq:H2-est} implies that $\pi(X_\alpha)\xi$ is a Cauchy net for every $\xi $ in $\hi^{1/2}$. But the pointwise convergence of self-adjoint operators on a common core implies the strong resolvent convergence of such operators (Theorem VIII.25.(a) of \cite{RS}), thus $\pi$ can be continuously extended. Finally, since the rotation group acts on $H^{3/2}(S^1, G)$ by continuous operators (see Lemma A.3 of \cite{CDIT}) and since $LG$ is dense in $H^{3/2}(S^1, G)$, we have that $\pi$ is actually a positive energy representation since it is $\text{Rot} $-covariant.
\end{proof}



\begin{prop} \label{prop:exp}
	\cite{TL} Let $R_s = \exp_{\Diff}(sh)$ be a smooth diffeomorphism of $S^1$, with $h$ a smooth vector field of the circle. Set $R_h=\{R_s\}_{s \in \R}$. Then the exponential map $L \g_0 \rtimes \R h \to LG \rtimes R_h$ is well defined and continuous. Moreover, if $X_\alpha = R_\alpha . X $ then
	\begin{equation} \label{eq:exp}
			\exp_{LG \rtimes R_h}(X+ \alpha h) = \exp_{LG}(X_\alpha) R_\alpha \,.
	\end{equation}
\end{prop}
\begin{proof}
	To compute the exponential map, we fix $X+ \alpha h$ in $L\g_0 \rtimes \R h$ and look for $f \colon \R \to LG \rtimes R_h $ which satisfies $(X + \alpha h)f = \dot{f}$ and $f(0)=1$. We suppose $f$ to be of the form $f_t = \gamma^t R_{\phi(t)}$ with $\gamma $ in  $LG$. As a manifold, $LG \rtimes R_h$ is the product of $LG$ and $R_h$, thus $s \mapsto \exp_{LG}(sX)R_{s\alpha}$ is the integral curve for $X + \alpha h$ at the identity. Therefore, with the notation 
	$ \gamma_s (\theta) = \gamma(R_s^{-1}(\theta)) $ we have
	\begin{equation*}
		\begin{split}
			(X + \alpha h)f_t &= \frac{d}{ds}  \bigg\vert_{s=0} \exp_{LG}(sX)R_{s\alpha} \gamma^t R_{\phi(t)} =  \frac{d}{ds} \bigg\vert_{s=0}  \exp_{LG}(sX) (\gamma^t)_{s\alpha}  R_{s\alpha + \phi(t)} \\
			& = X \gamma^t R_{\phi(t)}+ \alpha \frac{d}{ds}\bigg\vert_{s=0}   (\gamma^t)_s  R_{\phi(t)}+ \alpha \gamma^t h R_{\phi(t)} \,, \\
			\dot{f}_t  &= \bigg(  \frac{d}{dt} {\gamma^t} \bigg) R_{\phi(t)} + \phi'(t)\gamma^t h R_{\phi(t)} \,,
		\end{split}
	\end{equation*}
	whence $\phi(t) = \alpha t$, and we must solve
		\begin{equation*}
	 \frac{d}{dt} {\gamma^t} = X \gamma^t + \alpha \frac{d}{ds}\bigg\vert_{s=0}  (\gamma^t)_s \,, \quad
	\gamma^0 = 1\,.
	\end{equation*}
Since the solution of this equation is $\gamma^t  = \exp_{LG}(tX)_{\alpha t} =  \exp_{LG}(tX_{\alpha t})$, the assertion follows.
\end{proof}

\begin{cor}
	$LG \rtimes R_h$ is a smooth Fréchet Lie group, and in particular the map
	\[
	LG \times R_h \to LG \,, \quad (\gamma, R_s) \mapsto \gamma_s \,,
	\]
	is smooth. 
\end{cor}
\begin{proof}
We first point out that the map $ L\g_0 \times \R h \to L\g_0  $ given by $ (X, sh) \mapsto X_s $ is smooth and it is also a local homeomorphism. Hence it suffices to notice that $LG \times R_h$ is connected, check that the map $\exp_{LG \rtimes R_h}$ is a local homeomorphism thanks to Proposition \ref{prop:exp} and apply Lemma 2.2.1. of \cite{TL}.
\end{proof}

\begin{cor}
 The following holds in $PU(\hi)$:
	\[
	e^{\pi(X +  i\alpha h)} = \pi(\exp_{LG \rtimes R_h}(X + \alpha h)) \,.
	\]
\end{cor}
\begin{proof}
	By the Trotter product formula and Proposition \ref{prop:exp} we have the following identities in $PU(\hi)$: 
	\begin{equation}
		\begin{split}
		e^{\pi(X + i\alpha h)} & = \lim_{n \to \infty} (e^{i\alpha \pi(h)/n} e^{ \pi(X/n)} )^n = \lim_{n \to \infty} \pi (R_{\alpha/n} \exp_{LG} (X/n))^n \\
		& = \lim_{n \to \infty} \pi ( \exp_{LG} (X_{\alpha/n}/n) R_{\alpha/n}   )^n = \lim_{n \to \infty} \pi ( \exp_{LG \rtimes R_h} ((X + \alpha h)/n)   )^n \\
		& = \pi(\exp_{LG \rtimes R_h} (X + \alpha h)  ) \,, 
		\end{split}
	\end{equation}
	where we used the identity $e^{iT(h)} = \pi(e^h)$ which holds in $PU(\hi)$.
\end{proof}

\begin{lem}
	Let $\pi \colon G \to PU(\hi)$ be a strongly continuous projective representation of a topological group $G$. Then the map
	\[
	G \times U(\hi) \to U(\hi) \,, \quad (g,u) \mapsto \pi(g)u\pi(g)^* \,,
	\]
	is well defined and strongly continuous.
\end{lem}
\begin{proof}
The map is clearly well defined, and if $g_\alpha$ converges to $g$ in $G$ then we can choose lifts $v_\alpha$ and $v$ of $\pi(g_\alpha)$ and $\pi(g)$ such that $v_\alpha$ converges to $v$ in $U(\hi)$. But in the unitary group the strong topology and the $*$-strong topology coincide, and multiplication is continuous on bounded sets by the uniform boundedness principle, so the assertion follows.
\end{proof}
\begin{remark} \label{rem:lift}
	A   continuous projective representation $\pi \colon G \to PU(\hi)$ can be naturally lifted to a   continuous unitary representation $\tilde{\pi}$ of $ \widetilde{G} = \big\{(g,u) \in G \times U(\hi) \colon \pi(g) = [u] \big\} $ given by $\tilde{\pi}(g,u)= u$.
\end{remark}

\begin{thm} \label{thm:cov}
If $\gamma$ is in $H^{3/2}(S^1, G)$ and $X$ is in $H^{3/2}(S^1, \g_0)$, then
\begin{equation} \label{eq:1}
\pi(\gamma) \pi(X) \pi(\gamma)^* = \pi(\text{\em Ad}(\gamma)X) + ic(\gamma, X) \,,
\end{equation}
for some  continuous real function $c(\gamma, X)$.	Moreover, if $\gamma$ is in $H^{5/2}(S^1, G)$ and $h$ is a real vector field with $|h|_{3/2} < + \infty$, then
\begin{equation} \label{eq:2}
\pi(\gamma)\pi(X+ih)\pi(\gamma)^* = \pi(\text{\em Ad}(\gamma)X) + iT(h) + \pi(h \dot{\gamma}\gamma^{-1}) + ic(\gamma, X) + ic(\gamma, h) 
\end{equation}
for some continuous real function $c(\gamma,h)$.
\end{thm}

\begin{proof}
	We first prove \eqref{eq:2} in the smooth case. By the previous propositions, if $\gamma$ is in $LG$ and $Y=X+ih$ is in $L\g_0 \rtimes i\R h$, then the following identities hold in $PU(\hi)$:
	\begin{equation}
		\begin{split}
		\pi(\gamma) e^{s \pi(Y)}\pi(\gamma)^* & = \pi(\gamma) \pi(\exp_{LG \rtimes R_h} (sY)) \pi(\gamma)^* \\
		&=  \pi(\gamma \exp_{LG \rtimes R_h} (sY) \gamma^{-1}) \\
		&=  \pi( \exp_{LG \rtimes R_h} (s\text{Ad}(\gamma)Y)) \\
		&= e^{ s \pi(\text{Ad}(\gamma)Y)} \,,
		\end{split}
	\end{equation}
	and consequently $\pi(\gamma) e^{s \pi(Y)}\pi(\gamma)^* = \lambda(s) e^{ s \pi(\text{Ad}(\gamma)Y))}  $ for some function $\lambda \colon \R \to \T$. But $\lambda \colon \R \to \T$ is a continuous homomorphism and therefore $\lambda(t)  = e^{i a t }$ for a unique real number $a=c(\gamma, Y)$. We point out that $\text{Ad}(\gamma)$ has to be intended as the adjoint action with respect to the semidirect product $LG \rtimes R_h$. Notice also that $c(\gamma, Y) $ is linear in $Y$, so we can write 
	\[
	c(\gamma, X+ih) = c(\gamma, X) + c(\gamma, h) \,,
	\]
	where we set $c(\gamma, h)= c(\gamma, ih)$ for simplicity. Therefore, the claimed expression follows by the Stone theorem and by using the product rule for the derivative on $1 = \gamma_s \cdot \gamma_s^{-1}$. 

Now we prove \eqref{eq:1} in the Sobolev case. Consider $(\gamma_\alpha, X_\alpha)$ in $LG \times L\g_0$ converging to $(\gamma, X)$ in $H^{3/2}(S^1, G) \times H^{3/2}(S^1, \g_0)$. We have that both $	\pi(\gamma_\alpha) e^{s \pi(X_\alpha)}\pi(\gamma_\alpha)^*$ and $e^{ s \pi(\text{Ad}(\gamma_\alpha)X_\alpha)} $ strongly converge to the corresponding terms in $\gamma$ and $X$. By the argument used before we have that $e^{i c(\gamma_\alpha, X_\alpha )}$ converges to  $e^{i c(\gamma, X )}$, that is $e^{i c(\gamma, X )}$ is continuous in $\gamma $ and $X$. But continuity is a local property and the exponential map has local left inverses, thus $ c(\gamma, X )$ is continuous and the first part of the theorem is proved. Now we prove \eqref{eq:2} in the Sobolev case. Consider $\gamma$ in $H^{5/2}(S^1, G)$ and $h$ real with $|h|_{3/2}< + \infty$. Notice that   $i\pi(h)$ and $\pi(h \dot{\gamma}\gamma^{-1})$ are both  essentially skew-adjoint. Consider now smooth approximating nets $\gamma_\alpha \to \gamma$, $X_\alpha \to X$ and $h_\alpha \to h$ as before. By the previous propositions, the approximating right hand side of \eqref{eq:2} minus $c(\gamma_\alpha, h_\alpha)$ converges in the strong resolvent sense to the corresponding term in $\gamma$, $X$ and $h$ since we have a net of skew-adjoint operators pointwise convergent on a common core. Similarly, $\pi(X_\alpha + ih_\alpha)$ converges in the strong resolvent sense to $\pi(X + ih)$ and therefore
\[
\pi(\gamma_\alpha)e^{s \pi (X_\alpha + ih_\alpha)} \pi(\gamma_\alpha)^* \to \pi(\gamma)e^{s \pi (X + ih)} \pi(\gamma)^* 
\] 
strongly for every $s$ in $\R$. By the argument used before we have that $e^{ic(\gamma, h)}$ is continuous and thus $c(\gamma, h)$ is continuous. The thesis is proved.
\end{proof}


\begin{cor}
The subspace $\hi^s \subseteq \hi $ is $H^{5/2}(S^1, G)$-invariant for $s \geq 0$. Moreover, for $s \geq {5/2}$ real and $n \leq \floor{s-1}$ integer,  the corresponding map $H^{s}(S^1, G) \times \hi^n \to \hi^n/\T$ is jointly continuous.
\end{cor}

\begin{proof}
Since $\mathcal{D}(u^*Au) = u^*\mathcal{D}(A)$ for every unitary $u$ and every self-adjoint operator $A$, then
\begin{equation}
\begin{split}
\mathcal{D}((1+d)^s) & = \pi(\gamma)^* \mathcal{D}((1+d-i\pi(\dot{\gamma}\gamma^{-1}) + c(\gamma,d))^s)  \\
& \subseteq \pi(\gamma)^* \mathcal{D}((1+d)^s) \,.
\end{split}
\end{equation}
Since $\mathcal{D}((1+d)^s)  = \hi^s$ for $s \geq 0$, the $\hi^s$-invariance follows. Now we prove the second statement, where we can suppose $n \geq 1$. By Proposition 1.5.3. of \cite{TL} we have $\Vert \pi(\gamma) \xi \Vert_n \leq (1+M_{n-1})^n \Vert \xi \Vert_n$, where $M_p = C |\gamma^{-1}\dot{\gamma}|_{p + 1/2} + |c(\gamma^{-1},d)|$ for some $C>0$, and the joint continuity can be proved as in \cite{TL}.
\end{proof}


\begin{thm} \label{thm:covexp}
	With the hypotheses of Theorem \ref{thm:cov}, we have
	\[
	c(\gamma,X) = -\ell \int_0^{2 \pi} \braket{\gamma^{-1}\dot{\gamma}, X} \frac{d\theta}{2 \pi} \,, \quad 	c(\gamma,h) = - \frac{\ell}{2} \int_0^{2 \pi} h \braket{\gamma^{-1}\dot{\gamma}, \gamma^{-1}\dot{\gamma} } \frac{d\theta}{2 \pi} \,.
	\]
\end{thm}
\begin{proof}
We follow Theorem 1.6.3. of \cite{TL}, skipping some computations for the sake of brevity. Consider a smooth loop $\gamma$ in $LG$ and a smooth real vector field $h$. For $Y$ in $L\g_0 \rtimes i \R h$ we have
\begin{equation} \label{eq:FE}
c(\gamma_1 \gamma_2, Y) =  c(\gamma_2, Y) + c(\gamma_1, \text{Ad}(\gamma_2)Y ) \,. 
\end{equation}
If $\gamma^t = \exp_{LG}(tX)$, then the map $t \mapsto c(\gamma^t, Y) $ is differentiable at $t=0$ since $ LG \times R_h \to LG $ is smooth. In particular, we have that 
\[
\partial_t \big\vert_{t=0} c(\gamma^t, Y) = \ell B(X,Y) \,,
\]
and so 
\[
\partial_t \big\vert_{t=0} c(\gamma^t, h) = 0\,.
\]
By using \eqref{eq:FE} we have that $c(\gamma^t,h)$ is differentiable everywhere, with
\[
\partial_t c(\gamma^t, Y) =  \ell B(X,Y)  +c(\gamma^t, [X,Y] )\,,
\]
or more compactly
\begin{equation} \label{eq:ODE}
	\dot{c}_t(Y) =   i_X \ell B(Y) - (X. c_t)(Y)   \,.
\end{equation}
We naturally expect the solution of the ODE to be given by the Duhamel formula
\begin{equation} \label{eq:Duh}
c(\gamma^t,Y) = \ell B(X, \text{Ad}(\gamma^t) \int_0^t\text{Ad}(\gamma^{-\tau})Y d \tau)= \ell B(X, \int_0^t\text{Ad}(\gamma^s)Y d s) \,.
\end{equation}
Using $\frac{d}{dt} \text{Ad}(\gamma_t)Y = [X, \text{Ad}(\gamma_t)Y]$, it is easy to verify that \eqref{eq:Duh} defines a $C^1(\R, (L\g \rtimes i \R h)^*)$ solution of \eqref{eq:ODE} with initial condition $c_0=0$. The solution is unique. Finally, one can use \eqref{eq:Duh} and Corollary 1.6.2. of \cite{TL} to obtain the claimed expressions in the smooth case. By the continuity of $c(\gamma, Y)$ shown in Theorem \ref{thm:cov} the thesis is proved.
\end{proof}

\begin{cor} \label{cor:rules}
	By repeating the proof of Theorem \ref{thm:cov}, one can show that if $\gamma$ is in $H^{3/2}(S^1, G)$ and $X$ is in $H^{3/2}(S^1, \g_0)$, then
	\begin{equation} \label{eq:1b}
	\pi(\gamma)^* \pi(X) \pi(\gamma) = \pi(\text{\em Ad}(\gamma^{-1})X) + ib(\gamma, X) \,,
	\end{equation}
	for some  continuous real function $b(\gamma, X)$.	Moreover, if $\gamma$ is in $H^{5/2}(S^1, G)$ and $h$ is a real vector field with $|h|_{3/2} < + \infty$, then
	\begin{equation} \label{eq:2b}
	\pi(\gamma)^*\pi(X+ih)\pi(\gamma) = \pi(\text{\em Ad}(\gamma^{-1})X) + iT(h) - \pi(h \gamma^{-1}\dot{\gamma}) + ib(\gamma, X) + ib(\gamma, h) 
	\end{equation}
	for some continuous real function $b(\gamma,h)$. In particular, by $b(\gamma, Y ) = c(\gamma^{-1}, Y)$ we have
		\[
	b(\gamma,X) = - \ell \int_0^{2 \pi} \braket{\dot{\gamma} \gamma^{-1}, X} \frac{d\theta}{2 \pi} \,, \quad 	b(\gamma,h) = - \frac{\ell}{2} \int_0^{2 \pi} h \braket{\dot{\gamma}\gamma^{-1}, \dot{\gamma} \gamma^{-1}} \frac{d\theta}{2 \pi} \,.
	\]
\end{cor}

\section{Relative entropy and QNEC}
Let $\M$ be a von Neumann algebra in standard form, and let $\varphi$ and $\psi$ be two faithful, normal and positive linear functionals on $\M$ represented by vectors $\xi$ and $\eta$ in the natural cone. The relative entropy is defined by \cite{ED}
\[
S(\varphi \Vert \psi) = - (\xi | \log \Delta_{\eta, \xi} \xi) \,,
\]
where the above scalar product has to be intended by applying the spectral theorem to the relative modular operator $\Delta_{\eta, \xi}$. The relative entropy is nonnegative, convex, lower semicontinuous in the $\sigma(\M_*, \M)$-topology and monotone increasing with respect to von Neumann algebras inclusions $\mathcal{N} \subseteq \M$ (Theorem 5.3. of \cite{OP}). If the relative entropy is finite, then 
\begin{equation} \label{eq:der}
	S(\varphi \Vert \psi) = i \frac{d}{dt}\varphi((D \psi \colon D \varphi)_t)  \bigg\vert_{t=0} = - i \frac{d}{dt}\varphi((D \varphi \colon D \psi )_t )   \bigg\vert_{t=0}  \,,
\end{equation}
where $(D \varphi \colon D \psi )_t = (D \psi \colon D \varphi)_t^*$ is the Connes cocycle. \\

We now denote by $\mathcal{A}_\ell = \{  \mathcal{A}_\ell(I) \}_{I \in \mathcal{K}}$ the conformal net associated to the level $\ell$  vacuum representation $\pi$ of some loop group $LG$. We are interested in computing
\begin{equation} \label{eq:ent}
	S(t) = S_{\mathcal{A}_\ell(t, + \infty)}(\omega_\gamma \Vert \omega) \,,
\end{equation}
where $\omega$ is the vacuum state represented by the vacuum vector $\Omega$ and $\omega_\gamma = \omega \cdot \text{Ad}\, \pi(\gamma)^*$ is represented by $\pi(\gamma)\Omega$. To this purpose, we define the support of a loop $\gamma$  by 
\[
\text{supp} \, \gamma = \overline{\{\theta \in S^1 \colon \gamma(\theta) \neq e \} } \,.
\] 
We also introduce the groups of continuously differentiable and piecewise smooth loops
\begin{equation}
	\begin{split}
	B(z_1, \dots, z_n) & = \big\{ \gamma \in C^1_{ps}(S^1, G) \colon \gamma(z_i) =e \,, \dot{\gamma}(z_i) = 0 \big\} \,,
	\end{split}
\end{equation}
where by piecewise smooth we mean that right and left derivatives always exist and that $\gamma$ is smooth except for a finite number of points. By standard arguments such functions are in $H^{s}(S^1, G)$ for $s<5/2$. Therefore, by Theorem \ref{thm:ext} we have that if $\gamma$ is a loop in $B(z,w)$ then in $PU(\hi)$ we have
\begin{equation} \label{eq:split}
	\pi(\gamma) = \pi(\gamma_{(z, w)} ) \pi(\gamma_{(w,z)}) \,,
\end{equation}  
with $\text{supp} \, \gamma_{(z, w)}$ contained in the closure of $(z,w)$ and similarly for $\gamma_{(w,z)}$. In particular $\pi(\gamma_{(z, w)})  $ is in $\mathcal{A}_\ell((z,w))$ and $\pi(\gamma_{(w,z)})  $ is in  $\mathcal{A}_\ell((w,z))$, where $(z, w)$ denotes the interval of $S^1$ obtained by moving counterclockwise from $z$ to $w$. 
We also recall that by Bisognano-Wichmann we have $\log \Delta = - 2 \pi D$ with  $D=-\frac{i}{2}(L_1 - L_{-1})$, that is $\log \Delta =T(\delta)$ with $\delta$ the vector field generating $\delta(t).u = e^{- 2 \pi t}u$. We are therefore interested in computing the vacuum expectation of 
\begin{equation}
	\begin{split}
	\pi(\gamma)^*T(\delta) \pi(\gamma) & = T(\delta) + i \pi (\delta \gamma^{-1} \dot{\gamma}) + b(\gamma, \delta)	\,.
	\end{split}
\end{equation}

\begin{lem}
	In every vacuum representation, for $X$ in $H^{3/2}(S^1, \g)$ we have  $(\Omega | \pi(h X) \Omega) =0$ for every vector field $h$ in $\mathfrak{sl}_2(\C)$.
\end{lem}
\begin{proof}
	We can suppose $X$ to be in $L\g$. Notice that $ L\g = \g + \text{Ran}\,d $, since an element $X$ in $L\g$ with null average is the derivative of some element $Y$ in $L \g$. Therefore, 
	\[
	\pi(h X) = \pi(h.Y) + \pi (h a) = [\pi(h), \pi(Y)] + \pi (h a) \,, \quad a \in \g \,,
	\]
	and the thesis follows as the vacuum is  $G$-invariant and M\"ob-invariant.
\end{proof}

\begin{prop} \label{prop:ansatz'}
	Let $\gamma$ be a loop in $H^3(S^1, G)$. Pick a not dense open interval $I$ of the circle and write $\gamma = \gamma_I  \gamma_{I'} $ as in  \eqref{eq:split}.  Set $\Delta^{it}_I = e^{it\pi(\delta_I)}$ and denote by $\delta_I$  the generator of dilations of the interval $I$. If $\dot{\gamma}$ vanishes on the boundary of $I$ then the Connes cocycle $(D \omega_\gamma \colon D \omega)_{t}$  of $\mathcal{A}_\ell(I)$ is given by 
	\begin{equation} \label{eq:ans'}
	(D \omega_\gamma \colon D \omega)_{t} = e^{it(a + c(\gamma_I, \delta_I))}e^{t(i \pi(\delta_I) + \pi(\delta_I \dot{\gamma}_I \gamma_I^{-1} )) }\Delta^{-it}_I \,,
	\end{equation}
	for some $a=a_\gamma$ in $\R$. In particular, $a$ depends only on the values of $\gamma$  at the boundary of $I$ and $a_\gamma =0$ if $\gamma(z)=e$ for $z$ in the boundary of $I$. 
\end{prop}
\begin{proof}
	First we check that $\delta_I \dot{\gamma}_I \gamma_I^{-1} $ is in $H^2(S^1, G)$ since it vanishes with its first derivative on the boundary of $I$. Hence the right hand side of \eqref{eq:ans'}, which we denote by $u_t$, is a well defined unitary operator which is in $\mathcal{A}_\ell(I)$ by the Trotter product formula. To prove the existence of $a$ in $\R$ as in the statement it suffices to check that $u_t$ verifies the relations
	
(i) $\sigma^\gamma_t(x)= u_t \sigma_t (x) u_t^*\,, \quad x \in \mathcal{A}_\ell(I) \,,$

(ii) $u_{t+s}=u_t\sigma_t(u_s)\,.$
	
	Here $\sigma_t $ and $\sigma^\gamma_t $ are the modular automorphisms associated to the states $\omega$ and $\omega_\gamma$. The first relation follows by noticing that 
\begin{equation} \label{eq:computation}
	\begin{split}
	\sigma^\gamma_t(x) & = \text{Ad}\, \Delta^{it}_{I,\gamma} (x) = \text{Ad}\, \pi(\gamma)\Delta_I^{it} \pi(\gamma)^* (x) \\
	& = \text{Ad}\, \pi(\gamma)\Delta_I^{it} \pi(\gamma)^* \Delta_I^{-it} \Delta_I^{it} (x) \\
	& =  \text{Ad}\, u_t \cdot \sigma_t (x) \,,
	\end{split}
	\end{equation}
	where we used Lemma 3.(ii) of \cite{LP} and Theorem \ref{thm:covexp}. The second relation can be easily verified and thus $a$ does exist. Now we prove that $a=a_\gamma$ depends only on the values of $\gamma$  at the boundary of $I$. Consider $\eta$ in $H^3(S^1, G)$ such that, $\eta(z) = e$ and $\dot{\eta}(z)=0$ for $z$ in the boundary of $I$. Notice that $ (D \omega_{\eta \gamma} \colon D \omega)_{t} = \pi(\eta) 	(D \omega_{\gamma} \colon D \omega)_{t} \sigma_t(\pi(\eta)^*) $. Therefore, with the notation of Corollary \ref{cor:rules} we have 
\begin{equation}
\begin{split}
a_{\eta \gamma} - b(\eta_I \gamma_I, \delta_I) & = -i \frac{d}{dt}\omega_{\eta \gamma} ((D \omega_{\eta \gamma} \colon D \omega)_{t} ) \bigg\vert_{t=0}  = -i \frac{d}{dt}\omega_{\eta \gamma} ( \pi(\eta)	(D \omega_{\gamma} \colon D \omega)_{t} \sigma_t(\pi(\eta)^*) ) \bigg\vert_{t=0}  \\
& = a_\gamma  - b( \gamma_I, \delta_I) -  b(\eta_I \gamma_I, \delta_I) + b( \gamma_I, \delta_I) \,,
\end{split}
\end{equation}
and by the identity $a_{\eta \gamma} = a_\gamma$ the assertion is proved. If  $\gamma(z)=e$ for $z$ in the boundary of $I$ then $\pi(\gamma_I)$ is in $\mathcal{A}_\ell(I)$ and the last statement follows by Lemma 3.(iv) of \cite{LP}.
\end{proof}

Now we arrive to the main theorem of this work, that is we will use the previous results to prove the QNEC \cite{CF} on the loop groups models for the states of the type $\omega_\gamma$. Namely, we want to prove that the relative entropy \eqref{eq:ent} verifies
\begin{equation} \label{eq:Q}
S''(t) \geq 0 \,.
\end{equation}

\begin{thm} \label{thm:QNEC'}
	Let $\gamma$ be a loop in $H^3(S^1, G)$. Then the relative entropy \eqref{eq:ent} is finite and given by 
	\begin{equation} \label{eq:entropy'}
	S(t) = - \frac{\ell}{2} \int_t^\infty (u-t) \braket{\dot{\gamma} \gamma^{-1}, \dot{\gamma} \gamma^{-1}} {d u} \,.
	\end{equation}
	In particular, the QNEC \eqref{eq:Q} holds, with \begin{equation} 
	S''(t) = - \frac{\ell}{2} \braket{\dot{\gamma} \gamma^{-1}, \dot{\gamma} \gamma^{-1}}(t)  \geq 0 \,.
	\end{equation}
\end{thm} 
\begin{proof}
	Since the vacuum is $G$-invariant then we can replace $\gamma$ with $\gamma g$ for any $g$ in $G$, thus we can suppose $\gamma(\infty) =e$. We first suppose $\gamma$ to be in $B(\infty)$. We point out that if $\gamma(t) =e$ and $\dot{\gamma}(t) =0$ then $S(t)$ is finite and given by \eqref{eq:entropy'} since we can use equation \eqref{eq:der} and Proposition \ref{prop:ansatz'}. Now we prove that  $S(t)$ is finite for any $t$ real. Indeed, for any $t$ real we can pick a smooth loop $\eta$ with  $ \text{supp} \, \eta \leq t$ and such that $ \eta(t-k) = \gamma(t-k)^{-1} $ and $\dot{(\eta \gamma)}(t-k) =0$ for some $k>0$. This implies that
	\[
	S_{\mathcal{A}_\ell(t, + \infty)}(\omega_\gamma \Vert \omega) = S_{\mathcal{A}_\ell(t, + \infty)}(\omega_{ \eta \gamma} \Vert \omega) \leq S_{\mathcal{A}_\ell(t-k, + \infty)}(\omega_{ \eta \gamma} \Vert \omega) < + \infty \,,
	\]
	where the last relative entropy is finite by the argument used above. By similar arguments we have  that $ \bar{S}(t) = S_{\mathcal{A}_\ell(- \infty, t)}(\omega_\gamma \Vert \omega)  $ is finite for any $t$ real. Now we focus on the case $t=0$, since the general case follows by covariance. We suppose $\dot{\gamma}(0)=0$ and we write $\gamma = \gamma_+  \gamma_- $, with $\gamma_+(u)=e$ for $u \leq 0$ and $\gamma_-(u)=e$ for $u \geq 0$. By Proposition \ref{prop:ansatz'} we have 
	\[
	S(0) = a_\gamma - \frac{\ell}{2} \int_0^\infty u \braket{\dot{\gamma} \gamma^{-1}, \dot{\gamma} \gamma^{-1}} {d u} \,.
	\]
	Now we emulate some techniques used in \cite{RC} and we prove that $a_\gamma = 0$. Given $\lambda>0$ real, consider the function $f(u) = ue^{\lambda u}$. For $n>0$ integer, we consider a smooth diffeomorphism $\rho = \rho_{\lambda,n}$ of the circle such that, in the real line picture, it verifies $\rho(u) = f(u)$ for $0 \leq u \leq n - \frac{1}{n}$ and  $\rho(u) =f'(n)u+(f(n)-nf'(n))$ for $u \geq n$. We also suppose $\rho(u)/\rho'(u)$ to be uniformly bounded for $n - \frac{1}{n} \leq u \leq n$. Consider now the loop  $\gamma_{\lambda,n}(u) = \gamma(\rho^{-1}_{\lambda,n}(u))$. By the identity $a_\gamma = a_{\gamma_{\lambda,n}} $ and by monotone convergence we have 
	\begin{equation}
		\begin{split}
		0 & \leq \inf_\lambda  S_{\mathcal{A}_\ell(0, + \infty)}(\omega_{\gamma_{\lambda,n}} \Vert \omega)  = a_{\gamma} - \frac{\ell}{2}  \int_n^\infty (u-n) \braket{\dot{\gamma} \gamma^{-1}, \dot{\gamma} \gamma^{-1}} {d u}  \,,
		\end{split}
	\end{equation}
and by monotone convergence we have $a_\gamma \geq 0$. Now we prove the other inequality. Consider smooth paths $ \zeta_n$ in $G$ with extremes $\zeta_n(0) = e$ and $\zeta_n(1) = \gamma(0)$. We also suppose that $\dot{\zeta}_n(0) = \dot{\zeta}_n(1) = 0$, $\ddot{\zeta}_n(0)=0$ and $\ddot{\zeta}_n(1)=\ddot{\gamma}(0)/n^2$. Lastly, we suppose that $\zeta_n \to \zeta$ uniformly with its derivatives for some smooth path $\zeta$ on $G$. We now define
\[
\gamma_n(u) = \begin{cases}
\gamma(u) & u \geq 0 \,, \\
\zeta_n(nu+1) &  -1/n \leq u \leq 0 \,, \\
e & u \leq -1/n \,.
\end{cases}
\]
By monotonicity $ S_{\mathcal{A}_\ell(0, + \infty)}(\omega_{\gamma} \Vert \omega)  = S_{\mathcal{A}_\ell(0, + \infty)}(\omega_{\gamma_n} \Vert \omega) \leq S_{\mathcal{A}_\ell(-1/n, + \infty)}(\omega_{\gamma_n} \Vert \omega) $, and after a limit we have
\[
a_\gamma \leq - \frac{\ell}{2} \int_0^1 u \braket{\dot{\zeta} \zeta^{-1},\dot{\zeta} \zeta^{-1}} du \,.
\]
However, if we now consider the function $g_\lambda(u) = ue^{\lambda (u-1)}$ and we define $\zeta_\lambda(u) = \zeta(g_\lambda^{-1}(u))$, then 
\[
a_\gamma \leq - \frac{\ell}{2} \int_0^1 u \braket{\dot{\zeta_\lambda} \zeta_\lambda^{-1},\dot{\zeta_\lambda} \zeta_\lambda^{-1}} du \leq - \frac{\ell}{2\lambda } \int_0^1 u \braket{\dot{\zeta} \zeta^{-1},\dot{\zeta} \zeta^{-1}} du \to 0 \,, \qquad \lambda \to + \infty \,.
\]
Finally, we have proved that $a_\gamma =0$ if $\dot{\gamma}(0) =0$. To remove this condition, we notice that if $P$ is the generator of translations then the average energy in the state $\omega_\gamma$ is finite and given by 
\begin{equation} \label{eq:en}
E_\gamma=(\pi(\gamma)\Omega | P \pi(\gamma)\Omega) = - \frac{\ell}{2} \int_{- \infty}^{+ \infty}  \braket{\dot{\gamma} \gamma^{-1}, \dot{\gamma} \gamma^{-1}} \frac{du}{2 \pi} \,.
\end{equation}
Therefore we can apply Lemma 1. of \cite{CF}, namely for every $t_1$ and $t_2$ in $\R$ we have
\begin{equation} \label{eq:l1}
(S(t_1) - S(t_2)) + (\bar{S}(t_2)-\bar{S}(t_1)) = (t_2 - t_1)2 \pi E_\gamma \,.
\end{equation}
This implies that $S(t)$ and $\bar{S}(t) $ are both Lipschitz functions. Consider now a smooth real function $\rho(u)$ defined on $[0,1]$ and such that $\rho(0)=0$ and $\rho(1)=1$. We also suppose $\rho'(0) = \rho''(0)=0$,  $\rho'(1)=1$ and $\rho''(1)=0$. We define
\[
\gamma_n(u) = \begin{cases}
\gamma(u) & u \geq 1/n \,, \\
\gamma(\rho(nu)/n) &  0 \leq u \leq 1/n \,, \\
\eta(u) & u \leq 0\,,
\end{cases}
\]
where $\eta$ is a smooth function such that $\gamma_n$ is in $H^3(S^1, G)$. Therefore, by  \eqref{eq:l1} we have 
\[
0 \leq  S_{\mathcal{A}_\ell(0, + \infty)}(\omega_{\gamma_n} \Vert \omega)  - S_{\mathcal{A}_\ell(1/n, + \infty)}(\omega_{\gamma_n} \Vert \omega)  \leq \frac{2 \pi}{n} E_{\gamma_n} \to 0 \,,
\]
and thus we have
\begin{equation} \label{eq:f}
	\begin{split}
	S_{\mathcal{A}_\ell(0, + \infty)}(\omega_{\gamma} \Vert \omega)  & = \lim  S_{\mathcal{A}_\ell(1/n, + \infty)}(\omega_{\gamma} \Vert \omega) = \lim S_{\mathcal{A}_\ell(1/n, + \infty)}(\omega_{\gamma_n} \Vert \omega) \\
	& = \lim S_{\mathcal{A}_\ell(1/n, + \infty)}(\omega_{\gamma_n} \Vert \omega)  -  S_{\mathcal{A}_\ell(0, + \infty)}(\omega_{\gamma_n} \Vert \omega)  +  S_{\mathcal{A}_\ell(0, + \infty)}(\omega_{\gamma_n} \Vert \omega) \\
	& =  \lim S_{\mathcal{A}_\ell(0, + \infty)}(\omega_{\gamma_n} \Vert \omega)  \\
	& = - \frac{\ell}{2} \int_0^\infty u \braket{\dot{\gamma} \gamma^{-1}, \dot{\gamma} \gamma^{-1}} {d u} \,.
	\end{split}
\end{equation}
The most of the work is done. Now we just have to remove the condition $\dot{\gamma}(\infty) =0$. If we apply covariance to equation \eqref{eq:f} then we have 
\[
	S_{\mathcal{A}_\ell(- \infty,0)}(\omega_{\gamma} \Vert \omega)  = - \frac{\ell}{2} \int_{-\infty}^0 u \braket{\dot{\gamma} \gamma^{-1}, \dot{\gamma} \gamma^{-1}} {d u} 
\]
for any $\gamma$ in $H^3(S^1, G)$ such that $\dot{\gamma}(0)=0$. But this condition can be removed as in \eqref{eq:f}, and by covariance we have that the above expression of $S(0)$ holds for a generic loop $\gamma$ in $H^3(S^1, G)$. 
\end{proof}

	As a remark, we add a physical characterization of the second derivative $S''(t)$ which motivates this axiomatic definition of the QNEC. Since the quantity $E=E_\gamma$ given by \eqref{eq:en} is the average energy in the state $\omega_\gamma$, then we can define the density of energy
	\[
	E(t) = - \frac{\ell}{4 \pi}  \braket{\dot{\gamma} \gamma^{-1}, \dot{\gamma} \gamma^{-1}}(t) \,.
	\]
	In particular, we have that the QNEC is satisfied with the equality
	\[
	 E(t) = S''(t)/2 \pi \geq 0 \,.
	\]

\begin{cor}
	Let $\gamma$ be a loop in $H^3(S^1, G)$. Then
	\begin{equation} \label{eq:ent_r}
	S_{\mathcal{A}_\ell(- r, r)}(\omega_\gamma \Vert \omega) = - \frac{\ell}{2} \int_{-r}^r \frac{1}{2r}(r-u)(r+u) \braket{\dot{\gamma} \gamma^{-1}, \dot{\gamma} \gamma^{-1}} {d u} \,. 
	\end{equation}
	In particular, if $E_\gamma$ is the energy \eqref{eq:en} then we have the Bekenstein Bound
	\begin{equation} \label{eq:BB}
	S_{\mathcal{A}_\ell(- r, r)}(\omega_\gamma \Vert \omega)  \leq \pi r E_\gamma \,.
	\end{equation}
\end{cor}
\begin{proof}
	Since we proved Theorem \ref{thm:QNEC} for a generic loop in $H^3(S^1, G)$ then we can just apply M\"ob-covariance to equation \eqref{eq:entropy} to have the identity  \eqref{eq:ent_r}. See also \cite{ED} for some similar results.
\end{proof}

\subsection{QNEC on $LSU(n)$}

In this section we focus on the case $G=SU(n)$ and we use a construction illustrated in \cite{WA} to show that a positive energy representation of $LSU(n)$ can be extended to a positive energy representation of the Sobolev group $H^s(S^1, SU(n))$ for $s>1/2$. In particular, we will use this fact to provide a simpler proof of the QNEC \eqref{eq:Q}. \\


We begin by considering the natural action of $G=SU(n) $ on $V=\C^n$ and set  $H=L^2(S^1, V)$, or equivalently $H=L^2(S^1) \otimes V$. We can naturally define a continuous action $M$ of $LG$ on $H$ by $M_\gamma f(\phi) = \gamma(\phi)  f(\phi) $. We can also define an action of $\text{Rot} $ on $H$ by $R_\theta f(\phi) = f(\phi - \theta)$ with respect to the representation of $LG$ is covariant, that is it satisfies $R_\theta M_\gamma R_\theta^{-1} = M_{R_\theta \gamma}$. If $P$ is the orthogonal projection onto the Hardy space $H_+$, namely
\[
H_+ = \biggl\{  f \in L^2(S^1, V) \colon f (\theta) = \sum _{k \geq 0} f_k e^{ik\theta} \text{ with } f_k \in V \biggr\}\,,
\]
then we can define a new Hilbert space $H_P$ which is equivalent to $H$ as real Hilbert space, but with complex structure given by $J=iP-i(1-P)$. The Segal quantization criterion, which we now recall, allows us to define a positive energy representation of $LG$ on the fermion Fock space $\mathcal{F}_P = \Lambda H_P$ known as the {\em fundamental representation} of $LSU(n)$. Since $\mathcal{F}_P(0) = \Lambda V$, the fundamental representation of $LSU(n)$ is the direct sum of all the $n+1$  irreducible positive energy representations of $LSU(n)$ of level $\ell = 1$. The fundamental representation contains the {\em basic representation}, that is the unique level one vacuum representation.

\begin{defn}
	The \emph{restricted unitary group} is the topological group
	\[
	U_P(H)= \bigl\{ u\in U(H) \colon [u,P]\in L^2(H) \bigr\}\,.
	\]
	The  considered topology is the strong operator topology combined with the metric given by $d(u,v)=\Vert [u-v,P ]\Vert_2$.
\end{defn}

Any $u \in U(H)$ gives rise to an automorphism of $CAR(H)$, called \emph{Bogoliubov automorphism}, via $a(f) \mapsto a(uf)$. For every projection $P$ on $H$ there is an irreducible representation of $CAR(H)$ on $\mathcal{F}_P$ which is denoted by $\pi_P$. The Bogoliubov automorphism is said to be \emph{implemented} on $\mathcal{F}_P$ if $\pi_P(a(uf))=U\pi_P(a(f))U^*$ for some unitary $U \in U(\mathcal{F}_P)$. 

\begin{thm} {\bf Segal's quantization criterion.} \cite{WA} If $[u,P]$ is a Hilbert-Schmidt operator then $u$ is implemented on $\mathcal{F}_P$ by some unitary operator $U_P$. Moreover, $U_P$ is unique up to a phase and the constructed map $U_P(H) \to PU(\mathcal{F}_P)$ is continuous.
\end{thm}

\begin{prop}
	The fundamental representation of $LSU(n)$ can be extended to $H^s(S^1, SU(n))$ for any $s>1/2$. In particular, every  positive energy representation of $LSU(n)$ extends to a positive energy representation of $H^s(S^1, SU(n))$ for $s>1/2$.
\end{prop}
\begin{proof}
	Notice that since a loop $\gamma$ in $LSU(n)$ is also a map from $S^1$ to $M_n(\C)$ then we can write $\gamma$ as a Fourier series $\gamma(z)=\sum \widehat{\gamma}_kz^k$, where $\widehat{\gamma}_k \in M_n(\C)$. We consider on $H$ the basis $e_j^k(z)=z^ke_j$, where $(e_j)$ is the standard basis of $\mathbb{C}^n$. We define $M_{pq}=\widehat{\gamma}_{p-q}$ and we note that $M_\gamma e^k_j=\sum_i M_{ik} e^i_j$, so that $(e^p_i,M_\gamma e^q_j)=(e_i, M_{pq} e_j)$. So $M_\gamma$ is represented by a $\mathbb{Z} \times \mathbb{Z}$ matrix $(M_{pq})$ of endomorphisms. We have
	\begin{equation*}
	\begin{split}
	\Vert [P,M_{\gamma}]\Vert_2^2 & = \sum_{p \geq 0,q<0} \Vert M_{pq} \Vert_2^2 + \sum_{p < 0,q \geq 0} \Vert M_{pq} \Vert_2^2 \\
	& = \sum_{k>0} k \Vert \widehat{\gamma}_k \Vert_2^2 - \sum_{k<0} k \Vert \widehat{\gamma}_k \Vert_2^2 \\
	& = \sum_{k \in \mathbb{Z}} |k| \Vert \widehat{\gamma}_k \Vert_2^2  \leq \sum_{k \in \mathbb{Z}}( 1+ |k|)^{2s} \Vert \widehat{\gamma}_k \Vert_2^2 \,, \\	
	\end{split}
	\end{equation*}
	for $s>1/2$. It is easy to verify that the map $\gamma \mapsto M_\gamma \in U_P(H)$ is continuous. We also have that the rotation group acts on $H^s(S^1, G)$ by continuous operators (see Lemma A.3 of \cite{CDIT}), and by $[R_\theta,P]=0$ we have that the projective representation of $\text{Rot} $ is actually a strongly continuous unitary representation. Therefore, the thesis follows by the Segal quantization criterion, the complete reducibility of positive energy representations (Thm. 9.3.1. of \cite{PS}), Proposition 2.3.3. of \cite{TL} and remarks below.
\end{proof}

\begin{prop} \label{prop:ansatz}
	Let $\gamma$ be a loop in $H^1(S^1, SU(n))$. Pick a not dense open interval $I$ of the circle and write $\gamma = \gamma_I  \gamma_{I'} $ as in \eqref{eq:split}. Then, in $PU(\hi)$ we have
	\begin{equation} \label{eq:ans}
	(D \omega_\gamma \colon D \omega)_{t} = \pi(\gamma_I \delta_I(t). \gamma_I^{-1}) \,,
	\end{equation}
	where $(D \omega_\gamma \colon D \omega)_t$ is the Connes cocycle of $\mathcal{A}_\ell(I)$ and $\delta_I(t)$ denotes the dilation associated to $I$.
\end{prop}
\begin{proof}
	First we check that $\gamma_I \delta_I(t). \gamma_I^{-1}$ is in $H^1(S^1, SU(n))$ since it is continuous on the boundary of $I$, hence the right hand side of \eqref{eq:ans} is well defined. With the same computations of Proposition \ref{prop:ansatz'} we have that $\sigma^\gamma_t(x) =  \text{Ad}\, \pi(\gamma_+ \delta_I(t). \gamma_+^{-1}) \cdot \sigma_t (x) $ for $x$ in $\mathcal{A}_\ell(I)$. Therefore, we have that $(D \omega_\gamma \colon D \omega)_{t} $ is equal to $\pi(\gamma_+ \delta_I(t). \gamma_+^{-1})$ up to a unitary $V$ in the commutant of $\mathcal{A}_\ell(I)$, but $(D \omega_\gamma \colon D \omega)_t $ and  $\pi(\gamma_+ \delta_I(t). \gamma_+^{-1})$ are both in $\mathcal{A}_\ell(I)$ and thus $V$ is a scalar.
\end{proof}

\begin{thm} \label{thm:QNEC}
	Let $\gamma$ be a loop in $H^1(S^1, SU(n))$. Suppose also that, in the real line picture, the support of $\gamma$ is bounded from below. Then the relative entropy \eqref{eq:ent} is finite and given by 
\begin{equation} \label{eq:entropy}
		S(t) = - \frac{\ell}{2} \int_t^\infty (u-t) \braket{\dot{\gamma} \gamma^{-1}, \dot{\gamma} \gamma^{-1}} {d u} \,.
	\end{equation}
In particular, the QNEC \eqref{eq:Q} holds, with \begin{equation} 
S''(t) = - \frac{\ell}{2} \braket{\dot{\gamma} \gamma^{-1}, \dot{\gamma} \gamma^{-1}}(t)  \geq 0 \,.
\end{equation}

\end{thm}

\begin{proof}
	Since the vacuum is $SU(n)$-invariant then we can replace $\gamma$ with $\gamma g$ for any $g$ in $SU(n)$, thus we can suppose $\gamma(\infty) =e$.
	As above,  if $\gamma(t) =e$ then $S(t)$ is finite and given by \eqref{eq:entropy}. We can prove that $S(t)$ is finite for any $t$ real as in Theorem \ref{thm:QNEC'}, and similarly we have that $\bar{S}(t) = S_{\mathcal{A}_\ell(- \infty, t)}(\omega_\gamma \Vert \omega)$ is finite for any $t$ real. If $P$ is the generator of translations then the average energy $E_\gamma$ in the state $\omega_\gamma$ is finite and given by equation \eqref{eq:en}. Therefore we can apply Lemma 1. of \cite{CF}, and equation \eqref{eq:l1} holds.  This implies that $S(t)$ and $\bar{S}(t) $ are both Lipschitz functions and in particular they are absolutely continuous. The next step is an estimate on $S'(t)$. For simplicity we focus on the case $t=0$ and we write $\gamma = \gamma_+  \gamma_- $ with $\gamma_+(u)=e$ for $u \leq 0$ and $\gamma_-(u)=e$ for $u \geq 0$. By  Proposition \ref{prop:ansatz} the Connes cocycle  $u'_s = (D \omega \colon D \omega_\gamma)_s$ on $\mathcal{A}_\ell(- \infty,0)$ is equal in $PU(\hi)$ to $\pi(\gamma_- \delta(-s).\gamma_-^{-1})^*$. But also the state $\omega_\gamma \cdot \text{Ad}(u'_s)^*$ verifies the finiteness conditions required to apply Lemma 1. and thus we have $- S'(0) \leq 2 \pi E_s $, where $E_s =( u'_s \pi(\gamma)\Omega | P u'_s \pi(\gamma)\Omega)$ for $s$ real. However, one can simply prove that
\[
\inf_s 2 \pi  E_s =  - \frac{\ell}{2} \int_0^{+\infty} \braket{\dot{\gamma} \gamma^{-1}, \dot{\gamma} \gamma^{-1}} {d u} \,.
\]
Therefore, by repeating the argument with any $t$ in $\R$ we have
\[
-S'(t) \leq - \frac{\ell}{2} \int_t^{+\infty} \braket{\dot{\gamma} \gamma^{-1}, \dot{\gamma} \gamma^{-1}} {d u} \,.
\]
Finally, if we define
\[
F(t) = - \frac{\ell}{2} \int_t^\infty (u-t) \braket{\dot{\gamma} \gamma^{-1}, \dot{\gamma} \gamma^{-1}} {d u} \,,
\]
then we can conclude that $S(t)=F(t)$ for any $t$ in $\R$. Indeed, if the support of $\gamma$ 
is compact then $H(t)=S(t) - F(t)$ is an absolutely continuous function with nonnegative derivative and going to $0$ as $|t| \to + \infty$. If the support of $\gamma$ is contained in $(k, +\infty)$ then by lower semicontinuity $S(t) \leq F(t)$ for every $t$ real, and we can similarly deduce that $H(t)=0$ for every $t$ real.
\end{proof}

\section{Solitonic representations from discontinuous loops}

In this section we follow \cite{NS} and we construct proper solitonic representations of the conformal net $\mathcal{A} = \{ \mathcal{A}(I) \}_{I \in \mathcal{K}}$ associated to some vacuum positive energy representation $\pi$  of a loop group $LG$. We denote by $U$ the  projective unitary continuous representation of $\Diff$ and we always suppose $G$ to be simple, compact and simply connected. We recall that $\mathcal{K}$ denotes the family of all the open, nonempty and not dense intervals of the circle.  Given $h$ in $G$, we define
\[
L_hG =  \big\{ \zeta \in C^\infty (\R, G) \colon \zeta(x)^{-1} \zeta(x+ 2 \pi) = h \big\} \,.
\]
We denote by $\mathcal{I}_\R$ the set of all the open, nonempty and not dense intervals of $ S^1 \setminus \{-1 \}$. We then define $\sigma_\zeta = \{ \sigma_\zeta^I \}_{I \in \mathcal{I}_\R}$ as the collection of maps given by 
\[
\sigma_\zeta^I \colon \mathcal{A}(I) \to B(\hi) \,, \quad \sigma^I_\zeta(x) = \text{Ad}\, \pi(\zeta_I)(x) \,,
\]
where $\zeta_I$ is a loop in $LG$ such that $\zeta_I (\theta)= \zeta(\theta)$ for $\theta$ in $I $ seen as a subinterval of $ (- \pi, \pi)$. 

\begin{prop}
	$\sigma_\zeta$ is an irreducible locally normal soliton with index $1$. 
\end{prop}
\begin{proof}
	Normality on each $\mathcal{A}(I)$ follows because on these local algebras $\sigma_\zeta$ is given by the adjoint action by a unitary operator. The compatibility property is clear, since if $I \subseteq J$ then $\pi(\zeta_I \zeta_J^{-1})$ is in $\mathcal{A}(I') = \mathcal{A}(I)'$. The index is $1$ since if $I$ is in $\mathcal{I}_\R$ then $\sigma_\zeta(\mathcal{A}(I)) = \mathcal{A}(I)$, and for the same reason we have that 
	\[
	\bigvee_{I \in \mathcal{I}_\R} \sigma_\zeta(\mathcal{A}(I)) = \bigvee_{I \in \mathcal{I}_\R} \mathcal{A}(I) = B(\hi)
	\]
	since the conformal net $\mathcal{A}$ is irreducible (see \cite{NS} for the definitions).
\end{proof}

Since $\pi$ is irreducible if and only if it is irreducible as a projective representation of $LG$, then $\sigma_\zeta$ is irreducible if $\pi$ is irreducible (see also Corollary 1.3.3. of \cite{TL}).  Notice that for $\zeta$ in $L_h G$ we have that $ \zeta_t(\phi) = \zeta(\phi) \zeta(\phi - t)^{-1} $ is in $LG$ for any $ t$ in  $ \R $. We now denote by $\text{Rot}^\sim$ the universal covering of $\text{Rot} \cong \T$, the group of rotations of the circle. If $R_t$ is the unitary representation of $\text{Rot}$ associated to $\pi$, then we can define $R^\zeta_t = \pi(\zeta_t) R_t$ in $PU(\hi)$ and notice that
$R^\zeta_t R^\zeta_s = R^\zeta_{t+s}$. However, in general $R^\zeta_{2 \pi}$ is not a scalar and therefore $\sigma_\zeta$ is not $\text{Rot}$-covariant but only locally $\text{Rot}^\sim$-covariant. We recall the definition.

\begin{defn}
	Let $\mathcal{G}$ be a topological group equipped with some homomorphism $ \iota \colon \mathcal{G} \to \Diff$.  We say that a soliton $\sigma = \{ \sigma_I \}_{I \in \mathcal{I}_\R }$ is {\em locally $\mathcal{G}$-covariant} if there is a  unitary projective continuous representation $U_\sigma$ of $\mathcal{G}$ which satisfies the following property: if $I$ is in $\mathcal{I}_\R$ and $V$ is a connected neighborhood of the identity in $\mathcal{G}$ such that $g. I$ is in $\mathcal{I}_\R$ for every $g $ in $V$, then $ \text{Ad} \, U_\sigma (g) \sigma_I(x) = \sigma_{\iota(g).I}\text{Ad} \,   U (\iota(g) )(x) $ for every  $x $ in $\mathcal{A}(I) $.
\end{defn}

We can notice that if $\zeta$ is in $L_gG $ and $\eta$ is in $L_hG $ then  $\zeta \eta^{-1}$ is in $L_{h^{-1}g}G$ if $h^{-1}g$ is in $Z(G)$. In particular, if $\zeta$ and $\eta$ are in both in $L_hG$ then $\zeta \eta^{-1}$ is in $LG$ and $\sigma_\zeta$ is unitarily equivalent to $\sigma_\eta$.

\begin{thm}
	Let $\pi$ be a vacuum positive energy representation of $LG$ of level $\ell \geq 1$. Given $\zeta$ in $L_hG$, the soliton $\sigma_\zeta$ extends to a DHR representation if and only if $h$ is central.
\end{thm}
\begin{proof}
First we suppose $h$ to be in $Z(G)$.	A quick computation shows that in this case  $R^\zeta_{2 \pi} = \pi(h)$. By the identity $\pi(h) e^{\pi(X)} \pi(h)^* = e^{\pi(X)}$ for any $X$ in $L \g_0$ we have that $R^\zeta_{2 \pi} $ is a scalar since $\pi$ is irreducible. This implies that $\sigma_\zeta$ is locally $\text{Rot}$-covariant and we have that $\sigma_\zeta$  can be extended to a locally normal DHR representation by using the arguments of Proposition 3.8. of \cite{NS}. Now we suppose $h$ to be not central. By absurd, $\sigma_\zeta$ extends to a DHR representation and thus it is  $\text{Rot}$-covariant. Denote by $U^\zeta_{\theta}$ the corresponding projective representation of the circle. Since $\sigma_\zeta$ is locally normal, by using the additivity property one can show that
\[
\text{Ad} U_\pi^\zeta \cdot \sigma_\zeta (x) = \text{Ad} U_{-\pi}^\zeta \cdot \sigma_\zeta (x)= \sigma_\zeta \cdot \text{Ad} R_{ -\pi }  (x) = \text{Ad} R_{- \pi }^\zeta \cdot \sigma_\zeta (x) \,, \quad x \in \mathcal{A}((0, \pi)) \,.
\]
It follows that $U^\zeta_\pi R^\zeta_\pi$ is in $\mathcal{A}((0, \pi))' $, and by similarly proceeding on $\mathcal{A}((-\pi, 0)) $ we have that $U^\zeta_\pi R^\zeta_\pi$ is in $\mathcal{A}((-\pi, 0))' $. Therefore $ R^\zeta_\pi = U^\zeta_{- \pi} $ up to  a scalar and thus $R^\zeta_{2 \pi}$ is a scalar. Now consider a maximal torus $T \subset G$ containing $h$. Since $T$ is connected, we can suppose that $\zeta(x)$ belongs to $T$ for any $x$ in $\R$, and by commutativity we have that $R^\zeta_{2 \pi} = \pi(h)$ in $PU(\hi)$. Therefore we have that $h$ is a noncentral element acting on $\hi$ as a scalar. If we now consider the kernel
\[
N = \{ g \in G \colon \pi(g) \in \T \} \,,
\]
then $N$ is a normal subgroup of $G$ which is not contained in the center. But $G$ is simple and connected, hence we have that $N=G$, which is an absurd.
\end{proof}

We conclude this last section  by studying the equivalence classes of the solitons constructed above. The DHR representations $\sigma_z$ with $z$ in $Z(G)$ correspond to inequivalent irreducible positive energy representations $\pi_z$ of the same level than $\pi$ (see Remark \ref{rem:lift} and Theorem 3.2.3. of \cite{TL}). Now we pick a maximal torus $T$ in $G$. Consider $\zeta $ in $L_s G$ and $\eta$ in $L_t G$ for some $s$ and $t$ in $T$. We can suppose  $\zeta $  and $\eta$ to be  both contained in $T$. It can be easily noticed that
\begin{equation*}
	\sigma_\zeta \cdot \sigma_\eta = \sigma_{\zeta \eta} \,, \quad \zeta \eta \in L_{st}G \,, \qquad \sigma_\zeta^{-1} = \sigma_{\zeta^{-1}} \,, \quad \zeta^{-1} \in L_{s^{-1}}G \,.
\end{equation*}
It follows that $\sigma_\zeta $ and $ \sigma_\eta $ are unitarily equivalent if and only if $s=t$, hence we have infinitely many inequivalent solitons. If we consider two maximal tori $T$ and $T' = gTg^{-1}$, then what we can say is that we have the identity
\[
\sigma_{g \zeta g^{-1}} = \text{Ad} \pi(g) \cdot \sigma_\zeta \cdot \text{Ad} \pi(g)^* \,.
\]






\section*{Acknowledgements} 
I thank Simone Del Vecchio for suggesting me the problem, and Daniela Cadamuro and Henning Bostelmann for the hospitality in the period spent in Leipzig. \\

\end{document}